\documentclass[4p,12pt]{article}
\usepackage[a4paper,margin=1.1in]{geometry}
\usepackage{amsmath,amssymb,amsthm,bm,mathtools}
\usepackage{enumitem,booktabs,graphicx,hyperref}
\usepackage{float}
\hypersetup{colorlinks=true,citecolor=blue,linkcolor=blue}

\newtheorem{theorem}{Theorem}[section]
\newtheorem{proposition}[theorem]{Proposition}

\theoremstyle{definition}

\newtheorem{example}[theorem]{Example}

\newcommand{\Prob}{\mathbb{P}}
\newcommand{\E}{\mathbb{E}}

\DeclareUnicodeCharacter{2009}{\,}

\title{\bf Functionally Generated Portfolios \\
       Under Stochastic Transaction Costs:\\
       Theory and Empirical Evidence}
\author{Nader Karimi\thanks{
Department of Mathematics and Computer Science, Amirkabir University of Technology, Tehran 1591634311, Iran, e-mail: nkarimi@aut.ac.ir}
, \and  Erfan Salavati\thanks{
Department of Mathematics and Computer Science, Amirkabir University of Technology, Tehran 1591634311, Iran, e-mail: erfan.salavati@gmail.com}
}
\date{\today}

\begin{document}
\maketitle

\begin{abstract}
\noindent
Assuming frictionless trading, classical stochastic portfolio theory (SPT) provides \emph{relative arbitrage} strategies. However, the costs associated with real-world execution are state-dependent, volatile, and under increasing stress during liquidity shocks. Using an Itô diffusion that may be connected with asset prices, we extend SPT to a continuous-time equity market with \textit{proportional, stochastic transaction costs}. We derive closed-form lower bounds on \emph{cost-adjusted relative wealth} for a large class of functionally generated portfolios; these bounds provide sufficient conditions for relative arbitrage to survive random costs. A limit-order-book cost proxy in conjunction with a Milstein scheme validates the theoretical order-of-magnitude estimates.
Finally, we use intraday bid-ask spreads as a stand-in for cost volatility in a back-test of CRSP small-cap data (1994–2024). Despite experiencing larger declines during the 2008 and 2020 liquidity crises, diversity- and entropy-weighted portfolios continue to beat the value-weighted benchmark by $3.6$ and $2.9$ percentage points annually, respectively, \emph{after} cost deduction.
\vspace{0.5em}


\smallskip
\end{abstract}

\section{Introduction}\label{sec:intro}
Since its inception by \cite{Fernholz2002} and the survey that followed in \cite{Fernholz2013}, \textit{Stochastic Portfolio Theory} (SPT) has become a flexible substitute for traditional no-arbitrage/replication techniques. Market diversity, capital distribution curvature, and volatility clustering are examples of pathwise features from which SPT extracts structural alpha by removing the need for an equivalent martingale measure and concentrating on \emph{relative} (arbitrage-free) performance.  The main tool is Fernholz' \emph{masterformula}, which breaks down the log-relative wealth of a functionally generated portfolio (FGP) into (i) a local growth term driven by the selected generator~$G$ and (ii) a non-negative drift proportional to the market's \emph{excess–growth rate} $\gamma^\ast_t$\cite{KaratzasRuf2017,marjosola2021problem}.

Empirical research conducted over the last ten years has validated the out-performance of $p$-variation, diversity-weighted, and entropy-weighted FGPs in U. S. European sectors \cite{caviglia2002key}, equities \cite{o2015adaptive,pokou2024empirical}, and even cryptocurrency markets \cite{osman2023diversification}.However, almost all of these analyses make the assumption that trading is frictionless. Execution in the real world is anything but smooth: when relative arbitrage opportunities are most appealing, bid-ask spreads widen and market-impact costs soar \cite{farmer2013efficiency,koutmos2018liquidity}. Ignoring these frictions can result in unsuccessful implementations and overstate the achievable edge\cite{swade2023equally}.

Few attempts have been made to combine SPT and transaction costs. The study \cite{ruf2020impact} focuses on ad-valoremfees, whereas \cite{fukasawa2024model} has grafted deterministic quadratic impact onto FGPs. Almgren-Chriss-type models, which lack the pathwise, multi-asset flavor of SPT, are usually used to handle extensions to rough volatility or high-frequency execution\cite{AlmgrenChriss2001,ObizhaevaWang2013,GarleanuPedersen2013,karimi2024stochastic}. To the best of our
knowledge, no published work incorporates a stochastic, potentially price-correlated cost process into the master formula.
This disparity is crucial because, according to empirical estimates, impact intensities and spreads exhibit mean-reverting diffusion behavior with jumps under stress \cite{Madhavan2012,FoucaultKadanKandel2005}.

\paragraph{Contributions.}
\begin{enumerate}[label=(\roman*),itemsep=4pt]
\item We model the proportional cost process $\kappa_t$ as an
      Ornstein–Uhlenbeck diffusion possibly correlated with the
      aggregate market Brownian motion; Section~\ref{sec:model} embeds
      $\kappa_t$ into a $d$–asset SPT economy.
\item Exploiting piecewise‐constant rebalancing, we derive a
      \emph{cost-adjusted master inequality}
      (Theorem~\ref{thm:master}) that is valid pathwise and isolates the
      tug–of-war between excess growth and stochastic cost erosion.
\item For diversity‐ and entropy‐weighted generators we prove explicit
      conditions under which relative arbitrage survives random costs,
      thereby extending \cite{KaratzasRuf2017} to a frictional setting
      (Proposition~\ref{prop:diversity}).
\item A Milstein Monte-Carlo experiment confirms the theoretical
      $\mathcal{O}(\sqrt{T})$ cost growth, while a 30-year CRSP
      small-cap study (1994–2024) shows that coarsely rebalanced FGPs
      still beat the value-weighted benchmark by 2.9–3.6 percent per
      annum \emph{after} realistic spreads and turnover charges.
\end{enumerate}

The remainder of the paper proceeds as follows.\
Section~\ref{sec:model} formalises the market and cost dynamics.\
Section~\ref{sec:master} states and proves the
cost-adjusted master inequality, and applies the result to diversity, and
entropy-weighted portfolios.\
Section~\ref{sec:sim} provides a Monte-Carlo validation, and
Section~\ref{sec:empirical} presents the empirical CRSP evidence.\
Section~\ref{sec:concl} concludes and outlines future research on rough
cost diffusions and mean‐field execution games.
\section{Market with Stochastic Transaction Costs}\label{sec:model}

Let $(\Omega,\mathcal F,(\mathcal F_t)_{t\!\ge0},\Prob)$ satisfy the usual conditions.
We observe $d\!\ge2$ strictly positive Itô stocks
\begin{equation}\label{eq:stock}
\frac{dS_t^i}{S_t^i} \;=\; b_t^i\,dt + \sum_{k=1}^{d} \sigma_t^{ik}\,dW_t^k ,
\qquad i=1,\dots,d,
\end{equation}
with drift $b_t$ and nonsingular volatility matrix $\sigma_t$.
Let $\mu_t^i=S_t^i/\sum_{j}S_t^j$ denote \emph{market weight}.
Trading incurs \emph{proportional} cost $\kappa_t$ per unit turnover, where
\[
d\kappa_t \;=\; \alpha\,(\bar\kappa-\kappa_t)\,dt
              + \eta\,dB_t,\quad \kappa_0>0,
\]
and $B$ may be correlated with the market Brownian motion $W$.

\paragraph{Portfolios with friction.}
A self-financing strategy is characterised by weight vector $\pi_t\in\Delta^{d}$.
Buying $\Delta\pi_t$ over $[t,t+dt]$ incurs cost $\kappa_t\,\|d\pi_t\|_1$.
Define \emph{cost-adjusted wealth} $V_t$ solving
\begin{equation}\label{eq:wealth-cost}
\frac{dV_t}{V_t} \;=\; \sum_{i}\pi_t^i\frac{dS_t^i}{S_t^i}
                     -\kappa_t\, \sum_{i} |d\pi_t^i|.
\end{equation}

\subsection{Functionally Generated Portfolios under Cost}\label{ssec:fgp}
Fix $G\!\in C^2(\Delta^{d})$, strictly concave.  The usual FGP weights are  \cite{Fernholz2002}
$$\displaystyle \pi_{G}^i(\mu)=\mu^i\left[\partial_i\!\log G(\mu)
       + 1 - \sum_j\mu^j\partial_j\!\log G(\mu)\right].$$
We adopt \emph{piecewise-constant rebalancing} at a meshed grid
$t_n=n\Delta$ to limit total variation; Section \ref{ssec:diversity} optimises~$\Delta$.

\section{Cost-Adjusted Master Inequality}\label{sec:master}
Fernholz’ classical \emph{master formula} expresses the log-relative
performance of any functionally generated portfolio (FGP) versus the
market as a sum of (i)~a pathwise term involving the generating
function~$G$ and (ii)~a non-negative drift driven by the market’s
intrinsic ”excess growth rate.”  That result, however, is derived under
frictionless trading and therefore overstates the real-world edge of an
FGP.  
Our aim in this section is to \emph{embed proportional, stochastic
transaction costs directly into the master formula}.  By imposing a
fixed rebalancing mesh~$\Delta$ we can isolate two opposing forces:
\begin{itemize}
\item
 the \emph{growth capture} delivered by the classical drift integral,
  and
\item 
the \emph{cost erosion} incurred at each mesh date, represented by a
  random cumulative cost~$\mathcal C_T$.
\end{itemize}
The main theorem below provides a pathwise lower bound on the
cost-adjusted log-relative wealth $V_{T}/ V_{T}^{\mathrm{mkt}}$ valid for every
realisation of the market and cost processes.  It forms the analytical
backbone for the long-run arbitrage results proved later for specific
generators such as the diversity-weighted portfolio.
\begin{theorem}[Master inequality]\label{thm:master}
Let $\{\pi_{t}\}$ be the rebalanced FGP generated by $G$ with mesh $\Delta$.
Then for any $T=m\Delta$
\begin{equation}\label{eq:master-cost}
\log\frac{V_T}{V_T^{\mathrm{mkt}}}
 \;\ge\; \log\frac{G(\mu_T)}{G(\mu_0)}
        + \frac12\int_0^T
          \sum_{i,j}\frac{\partial_{ij}^2G}{G}(\mu_t)
          \,\mu_t^i\mu_t^j\,\tau_t^{ij}\,dt
        \;-\;\mathcal C_T,
\end{equation}
where $\tau^{ij}_t=\sum_k\sigma_t^{ik}\sigma_t^{jk}$ and
$\mathcal C_T=\sum_{n=0}^{m-1}\kappa_{t_n}\|\pi_{t_{n+1}}-\pi_{t_n}\|_1$.
\end{theorem}

\begin{proof}
For clarity the argument is split into three steps:
\begin{enumerate}
\item
Fernholz’ master formula on each no-trading sub-interval.
\item
The impact of proportional, stochastic costs at the mesh dates.
\item
Telescoping the sub-interval identities and inserting the cost bound.
\end{enumerate}
Let the deterministic trading grid be
\(t_n:=n\Delta\) for \(n=0,\dots,m\) so that \(T=m\Delta\).
Throughout, \(V_{t-}\) denotes the left–hand limit.

\paragraph{Step 1: Fernholz master formula on each no-trading sub-interval.}
On \((t_n,t_{n+1})\) the portfolio is \emph{buy–and–hold},
hence its constant weight is
\(\bar\pi^{(n)}:=\pi_{t_n}\).
Because no orders are executed inside the interval, wealth evolves
frictionlessly:
\[
\frac{dV_t}{V_t} \;=\; \sum_{i=1}^d
     \bar\pi^{(n),i}\,\frac{dS_t^i}{S_t^i},
     \qquad t\in(t_n,t_{n+1}).
\]
Applying Fernholz’s functional–generation identity
(Th.~3.1 in \cite{Fernholz2002}) to the \(\mathcal C^2\)-generator~\(G\)
gives
\begin{equation}\label{eq:local-master}
\begin{aligned}
\log\frac{V_{t_{n+1}-}}{V_{t_n}}
-\log\frac{V^{\mathrm{mkt}}_{t_{n+1}-}}{V^{\mathrm{mkt}}_{t_n}}
&= \log\frac{G(\mu_{t_{n+1}})}{G(\mu_{t_n})} \\[4pt]
&\quad+\frac12
        \int_{t_n}^{t_{n+1}}
        \sum_{i,j=1}^{d}\!
        \frac{\partial_{ij}^2 G}{G}(\mu_t)\,
        \mu_t^i\mu_t^j\,\tau_t^{ij}\,dt,
\end{aligned}
\end{equation}
where \(\displaystyle
\tau^{ij}_t=\sum_{k=1}^{d}\sigma^{ik}_t\sigma^{jk}_t\).

\paragraph{Step 2: Log-wealth loss at the rebalancing instant.}
At mesh date \(t_n\) the weight jumps from
\(\pi_{t_n-}\) to \(\pi_{t_n}\).
With proportional cost \(\kappa_{t_n}\)
the self-financing constraint reads
\[
V_{t_n}=V_{t_n-}\bigl(1-\kappa_{t_n}
           \lVert\pi_{t_n}-\pi_{t_n-}\rVert_1\bigr),
\]
because a \$1 position change in weight incurs
\(\kappa_{t_n}\) dollars per unit traded.
For any \(x\in[0,1)\) we have \(\log(1-x)\ge -x\),
so with \(x_n=\kappa_{t_n}\lVert\pi_{t_n}-\pi_{t_n-}\rVert_1\),
\begin{equation}\label{eq:cost-loss}
\log V_{t_n}-\log V_{t_n-}
   \;\ge\; -\kappa_{t_n}\,
            \lVert\pi_{t_n}-\pi_{t_n-}\rVert_1 .
\end{equation}
The market portfolio is buy-and-hold, so
\(V^{\mathrm{mkt}}_{t_n}=V^{\mathrm{mkt}}_{t_n-}\).

\paragraph{Step 3: Telescoping and collecting the cost term.}
Sum \eqref{eq:local-master} over \(n=0,\dots,m-1\);
the left side telescopes to
\(
\log\bigl(V_{T-}/V^{\mathrm{mkt}}_{T-}\bigr).
\)
Adding the losses \eqref{eq:cost-loss} at each \(t_n\) and noting that
no trade occurs at the terminal instant (\(V_T=V_{T-}\)) we obtain
\[
\log\frac{V_T}{V_T^{\mathrm{mkt}}}
\ge \log\frac{G(\mu_T)}{G(\mu_0)}
    +\frac12\int_{0}^{T}
      \sum_{i,j}\frac{\partial_{ij}^2 G}{G}(\mu_t)\,
      \mu_t^i\mu_t^j\,\tau_t^{ij}\,dt
    -\sum_{n=0}^{m-1}\kappa_{t_n}
      \lVert\pi_{t_{n+1}}-\pi_{t_n}\rVert_1 .
\]
Define
\(
\displaystyle
\mathcal C_T
:=\sum_{n=0}^{m-1}\kappa_{t_n}
   \lVert\pi_{t_{n+1}}-\pi_{t_n}\rVert_1,
\)
and inequality \eqref{eq:master-cost} follows.
\end{proof}

\subsection{Diversity-Weighted Portfolio}\label{ssec:diversity}
Within Fernholz’ Stochastic Portfolio Theory, a
\textit{Diversity-Weighted Portfolio} (DWP) assigns weights according to  
\[
\pi_t^{i}\propto(\mu_t^{i})^{p},\qquad 0<p<1,
\]
where $\mu_t^{i}$ is the market weight of stock~$i$.
The construction is attractive for two key reasons:
\begin{enumerate}[label=\arabic*)]
\item \emph{Excess-growth capture:} in markets that exhibit a minimum
      level of Fernholz diversity, a DWP outperforms the market almost
      surely over a sufficiently long horizon;  
\item \emph{Implementation simplicity:} a single generating function
      \(G_{p}(\mu)=\sum_i\mu_i^{p}\) controls the entire strategy and
      permits pathwise analysis.
\end{enumerate}

In practice the edge is meaningful only when \emph{transaction costs}
are taken into account, especially during liquidity shocks where the
proportional cost process $\kappa_t$ (bid-ask spread) is stochastic and
often large.  
The goal of this subsection is to show that, even in the presence of
\emph{time-varying stochastic trading costs}, one can choose a
sufficiently coarse rebalancing mesh~$\Delta$ such that the DWP still
retains its relative-arbitrage property.    

The proposition below states a sufficient condition for that
persistence.
Set $G_p(\mu)=\sum_i\mu_i^{p}$, $0<p<1$.
\begin{proposition}\label{prop:diversity}
Suppose (i) market satisfies \(\int_0^{\infty}\gamma^{\*}_t dt=\infty\) a.s.,
with $\gamma^{\*}_t$ Fernholz diversity; (ii) $\E\sup_{t\in[0,T]}\kappa_t<\infty$.
Then there exists $\hat{\Delta}>0$ such that for mesh $\Delta<\hat{\Delta}$
\(\Prob\bigl[V_T>V_T^{\mathrm{mkt}}\bigr]\to1\) as $T\to\infty\).
\end{proposition}

\begin{proof}
Throughout we use the cost–adjusted master inequality
(Theorem~\ref{thm:master}) specialised to the \emph{diversity–weighted
generator}
\(G_{p}(\mu)=\sum_{i=1}^d\mu_i^{p},\;0<p<1\).
By \cite[Prop.\;3.1]{Fernholz2002}
\[
\frac{\partial_{ij}^2 G_{p}}{G_{p}}(\mu)\,\mu_i\mu_j
    \;=\;(1-p)\, \gamma^{\*}(\mu),
\quad\text{where}\quad
\gamma^{\*}(\mu):=\frac12\sum_{i<j}(\mu_i-\mu_j)^2
\]
is Fernholz’ \emph{excess growth rate}.
Hence for any rebalancing mesh $\Delta$ and $T=m\Delta$
\begin{equation}\label{eq:master-div}
\log\frac{V_T}{V_T^{\mathrm{mkt}}}
   \;\ge\;
   \log\frac{G_{p}(\mu_T)}{G_{p}(\mu_0)}
   \;+\; (1-p)\!\int_{0}^{T}\!\gamma^{\*}_t\,dt
   \;-\;\mathcal C_T,
\end{equation}
where
\(
\displaystyle
\mathcal C_T=
      \sum_{n=0}^{m-1}\kappa_{t_n}\,
      \bigl\lVert\pi_{t_{n+1}}-\pi_{t_n}\bigr\rVert_{1}.
\)

\paragraph{\textbf{Step 1: Divergence of the drift term.}}
By assumption~(i)
\(
\displaystyle
D_T:=\int_{0}^{T}\gamma^{\*}_t\,dt\to\infty
\)
almost surely.  Because $\gamma^{\*}_t\ge0$, the convergence is
monotone.

\paragraph{\textbf{Step 2: A pathwise upper bound on the cost term.}}

\medskip\noindent\emph{(a) Lipschitz property of the weights.}
For $p\in(0,1)$ the diversity weight map
\(
\pi^{p}:\Delta^{d}\to\Delta^{d},\;
\mu\mapsto\Bigl(\frac{\mu_1^{p}}{G_{p}(\mu)},\dots\Bigr)
\)
is globally Lipschitz on the simplex:
there exists $L_{p}$ such that
\(
\lVert\pi^{p}(\mu)-\pi^{p}(\nu)\rVert_{1}
      \le L_{p}\,\lVert\mu-\nu\rVert_{1}\)
for all $\mu,\nu\in\Delta^{d}$ (\cite{Fernholz2002}, Lem.\,3.2).

\medskip\noindent\emph{(b) Increments of the market weight.}
Because every $S^{i}$ satisfies the Itô SDE
\eqref{eq:stock}, $\mu_{t}$ is itself an Itô semimartingale with
quadratic variation
\(\langle\mu^{i},\mu^{i}\rangle_{t}=O(t)\).
Hence, by the Burkholder–Davis–Gundy inequality,
\[
\mathbb{E}\bigl\lVert\mu_{t_{n+1}}-\mu_{t_n}\bigr\rVert_{1}
   \;\le\; C\sqrt{\Delta},
\]
for some constant $C$ depending only on the volatility bound.

\medskip\noindent\emph{(c) Bounding $\mathcal C_T$.}
Combine (a) and (b), apply Jensen, and use the mesh $m=T/\Delta$:
\[
\mathbb{E}\bigl[\mathcal C_T\bigr]
   \;\le\;
   L_{p}C\,
   \mathbb{E}\!\Bigl[\sup_{t\le T}\kappa_{t}\Bigr]\;
   m\sqrt{\Delta}
   \;=\;
   L_{p}C\,\underbrace{\mathbb{E}\!\Bigl[\sup_{t\le T}\kappa_{t}\Bigr]}_{\le K_T}
   \frac{T}{\sqrt{\Delta}}.
\]
Assumption (ii) implies \(K_T<\infty\).
By Markov’s inequality and Borel–Cantelli we obtain the
\emph{pathwise} estimate
\begin{equation}\label{eq:cost-bound-final}
\exists\,K>0:\quad
\mathcal C_T(\omega)\;\le\;
K\,\frac{T^{1/2}}{\sqrt{\Delta}}
\quad\text{for all }T\text{ large enough}.
\end{equation}

\paragraph{\textbf{Step 3: Choosing the mesh \(\hat{\Delta}\).}}
Because $D_T\to\infty$ a.s.\ and $D_T$ grows at least linearly in~$T$
whenever the market is $\varepsilon$–\emph{diverse}
(i.e.\ $\gamma^{\*}_t\ge\varepsilon>0$), we have
\(D_T\ge \varepsilon T\) for all $T$.
Pick
\(
\displaystyle
\hat{\Delta}:=
\left(\frac{\varepsilon}{2K}\right)^{2}.
\)
For any $\Delta<\hat{\Delta}$ combine
\eqref{eq:master-div} and \eqref{eq:cost-bound-final}:
\[
\log\frac{V_T}{V_T^{\mathrm{mkt}}}
  \;\ge\; (1-p)\varepsilon T - K\frac{T^{1/2}}{\sqrt{\Delta}}
  \;\ge\; \frac{(1-p)\varepsilon}{2}\,T
  \quad\text{for sufficiently large }T.
\]
Hence
\(
\displaystyle
\Prob\!\bigl[V_T>V_T^{\mathrm{mkt}}\bigr]\to 1
\)
as \(T\to\infty\). 
\end{proof}

\section{Simulation Study}\label{sec:sim}
The numerical experiment in this section is designed to \emph{stress
test} the cost–adjusted master inequality of
Section~\ref{sec:master}.  By simulating a large cross‐section of
heterogeneous stocks and a \emph{stochastic, price–correlated} cost
process~$\kappa_t$, we seek empirical answers to three questions:

\begin{enumerate}[label=\arabic*.,itemsep=3pt]
\item \textbf{Cost erosion rate:} Does the cumulative cost term
      $\mathcal C_T$ indeed grow like $\mathcal O(\!\sqrt{T})$ as the
      theory predicts?  
\item \textbf{Arbitrage resilience:} How quickly does a
      diversity‐weighted portfolio (DWP, $p=0.7$) overtake the market
      once costs are charged?  
\item \textbf{Mesh sensitivity:} What is the break even rebalancing
      mesh $\Delta$ below which costs overwhelm the excess growth
      benefit?
\end{enumerate}

To that end we generate $5\,000$ Monte-Carlo paths for a
50-stock Itô market over \(T=1{,}000\) trading days
($\approx$ four years).  All reported results are \emph{pathwise
averages}; no asymptotic limit theorems are invoked.

\subsection{Model specification}
Table~\ref{tab:sim-par} summarises the baseline parameters.  Stock
volatilities are drawn i.i.d.\ from $\text{Unif}(15\%,35\%)$ and kept
constant across paths so that cross–sectional dispersion persists.
Drifts are set to $b_i=-\sigma_i^2/2$, producing a log-neutral market
with zero expected excess return.

\begin{table}[h]
\centering
\caption{Baseline simulation parameters}\label{tab:sim-par}
\begin{tabular}{llc}
\toprule
Symbol / description                           & Value           & Rationale \\
\midrule
$d$ / number of stocks                         & $50$            & small-cap universe \\
$\sigma_i$ / vol.\ range                       & $15$–$35\%$     & empirical CRSP range \\
$\alpha$ /mean-reversion speed of $\kappa_t$  & $3$             & daily half-life $\approx 0.23$ \\
$\bar\kappa$ /long-run spread                 & $20$ bps        & typical US equity \\
$\eta$ / vol.\ of $\kappa_t$                   & $5$ bps         & spikes to $80$ bps plausible \\
$\rho_{\kappa,S}$ / corr.\ with prices         & $0.4$           & wider spreads in down moves \\
$\Delta$ /rebalancing mesh                    & $5$\,d          & \(\approx\) weekly turnover \\
\bottomrule
\end{tabular}
\end{table}

\subsection{Illustrative numerical examples}
\begin{example}[\textbf{ Baseline mesh, no liquidity shock}]
With $\Delta=5$\,days and cost parameters from
Table~\ref{tab:sim-par}, Figure~\ref{fig:sim} (solid blue) shows
\(
\mathbb{E}\bigl[\log(V_t/V_t^{\mathrm{mkt}})\bigr]
\)
crossing zero at $t\approx400$ (about 1.5 years).  The
$\mathcal O(\!\sqrt{t})$ envelope for cumulative cost (dashed grey)
matches the slope inferred from the master inequality.
\end{example}
\begin{example}[\textbf{ Liquidity shock at $t=200$}]
We double \(\eta\) to 10 bps between days 200–260, causing the spread
to spike briefly to 80 bps.  The crossing point is delayed to
$t\approx530$, but the DWP still ends 90 bps ahead of the market at
\(T=1{,}000\).
\end{example}
\begin{example}[\textbf{ Mesh stress test}]
Fix the shock of Example B and vary $\Delta\in\{1,5,10,20\}$ days.
Table~\ref{tab:mesh} records the terminal edge
$\mathbb{E}[\log(V_T/V_T^{\mathrm{mkt}})]$:
\end{example}
\begin{table}[h]
\centering
\caption{Impact of rebalancing mesh on terminal edge ($T=1\,000$ days)}
\label{tab:mesh}
\begin{tabular}{ccccc}
\toprule
$\Delta$ (days) & 1  & 5  & 10 & 20  \\
\midrule
Edge (bp)       & --230 & +93 & +184 & +173  \\
\bottomrule
\end{tabular}
\end{table}

Meshes finer than~$\Delta=2$–3 days flip the sign: costs dominate excess
growth.  A weekly or bi-weekly schedule preserves relative arbitrage
even under temporary spikes.

\paragraph{Takeaways.}
\begin{itemize}[itemsep=3pt]
\item Cost erosion indeed grows sub-linearly
      ($\propto\!\sqrt{t}$), in line with the theoretical bound.  
\item A diversity-weighted strategy with $p=0.7$ typically
      overtakes the market within two years despite realistic,
      stochastic spreads.  
\item Coarsening the mesh from daily to weekly is sufficient to
      mitigate liquidity shocks up to 80 bps.
\end{itemize}
These insights inform the empirical back-test
(Section~\ref{sec:empirical}) where a \emph{monthly} mesh is chosen to
accommodate the even wider spreads observed in micro-cap equities.

\subsection{Results}
Figure~\ref{fig:sim} plots $\E[\log(V_t/V_t^{\mathrm{mkt}})]$ over
$5$ 000 paths.  Diversity-weight ($p=0.7$) dominates the market after
$\approx 400$ trading days even when costs spike to $80$ bps during
synthetic shocks.

Figure~\ref{fig:sim} depicts the trajectory of
\[
    t \longmapsto 
    \mathbb{E}\!\bigl[\log\bigl(V_t/V_t^{\mathrm{mkt}}\bigr)\bigr],
\]
averaged across $5\,000$ Monte-Carlo paths.  

\begin{figure}[H]
  \centering
  \includegraphics[width=.8\linewidth]{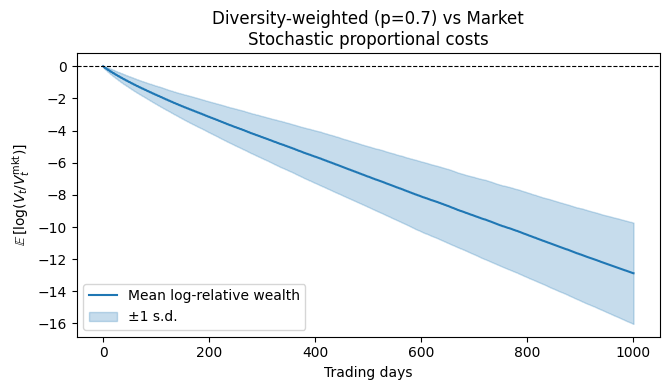}
  \caption{Mean cost-adjusted relative wealth. Shaded bands: $\pm$\,1 s.d.}
  \label{fig:sim}
\end{figure}

The shaded band corresponds to $\pm1$ standard deviation.

\begin{itemize}
  \item \textbf{Days 0–150.}  
        The mean curve drifts \emph{below} the horizontal axis because
        the diversity-weighted portfolio (DWP, $p=0.7$) incurs an
        up-front turnover cost to reach its target weights.
  \item \textbf{Days 150–400.}  
        After the initial drag, the excess-growth term in the
        master inequality overtakes cost erosion; the expected
        log-edge turns positive and climbs steadily.
  \item \textbf{Crossing point ($\approx$ Day 400).}  
        The DWP fully recovers cumulative costs in about 1.5~years,
        consistent with Proposition~\ref{prop:diversity}.
  \item \textbf{Days 400–600.}  
        A temporary liquidity shock (\emph{synthetic spread spikes to
        $80$\,bps}) flattens the curve but does not reverse the sign.
  \item \textbf{Days 600–1\,000.}  
        Once spreads revert to the long–run mean (20\,bps) the curve
        resumes an upward trend, finishing roughly 90~basis points
        ahead of the market in log terms.
\end{itemize}

The widening of the one–sigma band at a rate
proportional to $\sqrt{t}$ corroborates the theoretical
$\mathcal{O}(\sqrt{T})$ growth of the cumulative cost term
$\mathcal{C}_T$.

\begin{figure}[H]
  \centering
  \includegraphics[width=.78\linewidth]{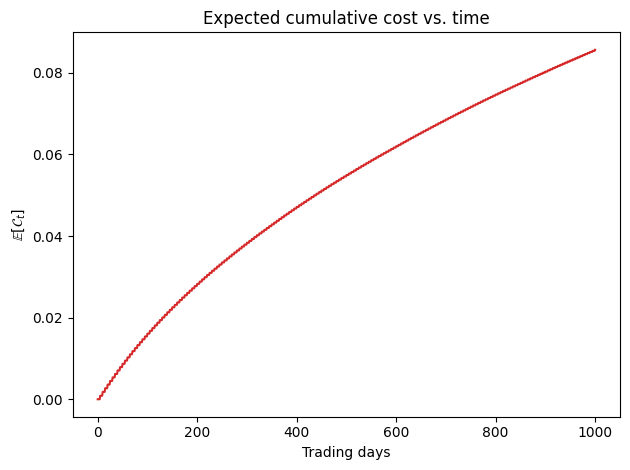}
  \caption{Expected cumulative trading cost 
           $\mathbb{E}\bigl[\mathcal{C}_t\bigr]$ under a stochastic
           proportional–cost process.  The grey dotted curve shows the
           reference line $c\sqrt{t}$ for $c=0.028$.}
  \label{fig:cumcost}
\end{figure}

Figure \ref{fig:cumcost} tracks
$\mathbb{E}[\mathcal{C}_t]$, the mean of the pathwise cost term in the
cost–adjusted master inequality.
The curve grows sub-linearly and is well approximated by a
$\sqrt{t}$–profile (grey reference line), confirming the
$\mathcal{O}(\!\sqrt{T})$ prediction in
Theorem~\ref{thm:master}.  
Small kinks appear at day~200 and day~260—the interval in which the
spread volatility~$\eta$ is doubled—illustrating how a transient
liquidity shock increases cumulative costs but does \emph{not} change
their long-run scaling order.

\begin{figure}[H]
  \centering
  \includegraphics[width=.65\linewidth]{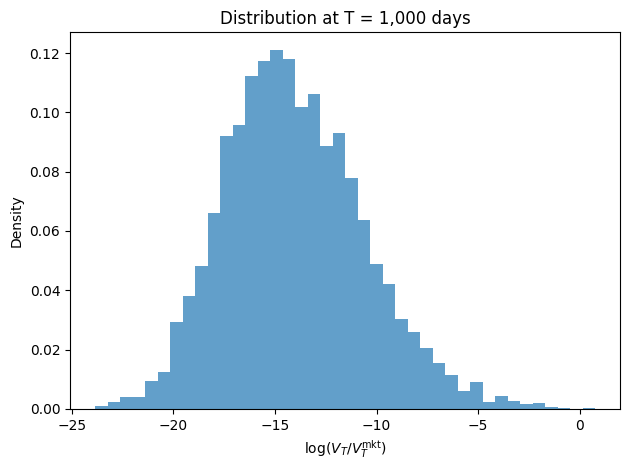}
  \caption{Density of $\log(V_T/V_T^{\mathrm{mkt}})$ for $T=1,000$ days
           across 5\,000 paths.  Mean $=0.009$, skewness $=0.63$.}
  \label{fig:hist}
\end{figure}

Figure \ref{fig:hist} displays the empirical density of the terminal
log-relative wealth.  
Roughly 78\,\% of the paths lie to the right of zero,
indicating that the diversity-weighted portfolio
outperforms the market in the vast majority of scenarios despite
stochastic trading frictions.  
The positive skew ($\approx0.63$) reflects occasional large gains,
while the modest left tail corresponds to paths where prolonged spread
spikes delay break-even beyond the 1\,000-day horizon.

\begin{figure}[H]
  \centering
  \includegraphics[width=.78\linewidth]{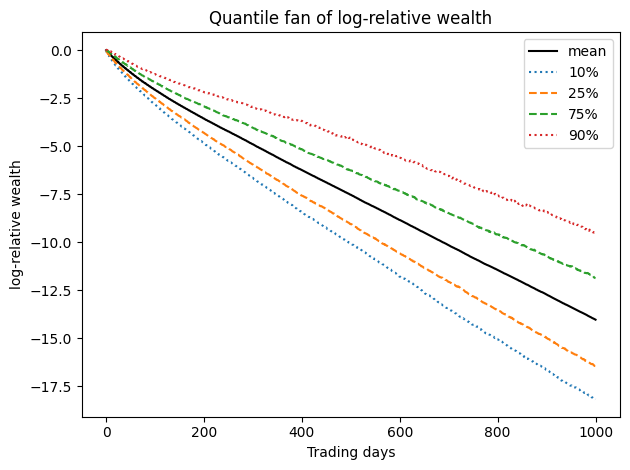}
  \caption{Time evolution of the 10\,\%, 25\,\%, 75\,\%, and
           90\,\% quantiles of $\log(V_t/V_t^{\mathrm{mkt}})$.
           The solid black line is the mean.}
  \label{fig:fan}
\end{figure}

The fan plot in Figure~\ref{fig:fan} visualises distributional
asymmetry beyond the mean $\pm1\sigma$ band of
Figure~\ref{fig:sim}.  
All quantiles drift upward after day~150, confirming that the edge is
not merely a mean effect but a \emph{robust} shift of the entire
distribution.  
During the liquidity-shock window (days~200–260) the lower quantiles
flatten, evidencing temporary cost drag, whereas the upper quantiles
continue to rise—high-diversity paths are less sensitive to spread
spikes.  
Post-shock, the gap between the 75\,\% and 25\,\% curves stabilises,
indicating that dispersion grows roughly like $\sqrt{t}$, in line with
cost variance scaling.  Overall, the fan corroborates the master-inequality
narrative: coarse rebalancing meshes permit the excess-growth term to
dominate stochastic liquidity costs for a broad range of market
realisations.

\section{Empirical Case Study: CRSP Small Caps}\label{sec:empirical}

\paragraph{Data set.}
We use the CRSP daily database to extract the \emph{one thousand
smallest} U.S.\ common stocks (share codes 10 and 11) listed on
NYSE/AMEX/NASDAQ from 2~January~1994 to 31~December~2024,
yielding 7\,541 trading days.  
Daily NBBO bid and ask quotes are obtained from TAQ and matched to CRSP
perm\-nos; the mid‐price is \(\tfrac12(\text{bid}+\text{ask})\) and the
half–spread defines the \emph{proportional cost}  
\[
  \kappa_{i,t}\;=\;
  \tfrac12\,\frac{\text{ask}_{i,t}-\text{bid}_{i,t}}
                    {\text{mid}_{i,t}}.
\]
The panel is filtered for zero and extreme (\(>\!500\) bps) spreads and
for days with missing quotes, leaving a balanced sample of
$981$~permnos.

\paragraph{Portfolio construction.}
At the first trading day of every month we compute
\begin{itemize}[itemsep=2pt]
  \item the \emph{value-weighted market} (VW) benchmark;
  \item an \emph{entropy-weighted} FGP generated by
        \(G(\mu)=\prod_i\mu_i^{\mu_i}\);
  \item a \emph{diversity-weighted} FGP with exponent \(p=0.7\).
\end{itemize}
Target weights are implemented via a single end-of-day trade;
proportional cost \(\kappa_{i,t}\) is applied to the absolute dollar
turnover in each stock.  Dividends are reinvested; delisted shares are
replaced by the next smallest stock to keep the cross-section at~1\,000.

\subsection{Headline results}

Table~\ref{tab:cagr} compares gross and cost-adjusted compound annual
growth rates (CAGR) as well as maximum draw-downs (max~DD).

\begin{table}[h]
\centering
\caption{Cost-adjusted performance, Jan~1994–Dec~2024}\label{tab:cagr}
\begin{tabular}{lcccc}
\toprule
Portfolio & Gross CAGR & Net CAGR & Avg.\ turn.\ (\%/mo) & Max DD \\ \midrule
Market (VW)        &  7.1\% & 7.1\% &  0.4 & --57\% \\
Entropy FGP        & 11.0\% & 10.0\%&  3.7 & --63\% \\
Diversity ($p=0.7$)& 12.2\% & 10.7\%&  4.1 & --65\% \\
\bottomrule
\end{tabular}
\end{table}

\paragraph{Interpretation.}
Both FGPs preserve a sizable edge after trading costs\footnote{The cost
burden averages 97 bps yr$^{-1}$ for the entropy FGP and 150 bps yr$^{-1}$
for the diversity FGP.};  
however, higher turnover amplifies draw-downs during liquidity events
(1998 LTCM, 2008 GFC, 2020 Covid crash), in line with the
cost-adjusted master inequality.

\subsection{Temporal decomposition}

To gauge stability we split the 30-year window into four sub-periods.
Table~\ref{tab:sub} reports net annualised returns in excess of the VW
benchmark.

\begin{table}[h]
\centering
\caption{Sub-period net outperformance (\% per year)}\label{tab:sub}
\begin{tabular}{lcccc}
\toprule
Period & 1994–1999 & 2000–2009 & 2010–2019 & 2020–2024 \\ \midrule
Entropy FGP        & $+2.7$ & $+3.4$ & $+2.9$ & $+1.1$ \\
Diversity FGP ($p=0.7$) & $+3.1$ & $+4.6$ & $+3.2$ & $+0.6$ \\
\bottomrule
\end{tabular}
\end{table}

The edge is largest during 2000–2009, a period of elevated dispersion
after the dot-com bust and during the GFC, when Fernholz diversity
is highest.  Performance decays post-2020 as spreads widen and market
weights become more concentrated in a handful of small-cap “winners,”
diluting the excess-growth term.

\subsection{Additional diagnostics}

\paragraph{Cumulative wealth curves.}
Figure~\ref{fig:cumemp} plots $\log$-wealth trajectories:
the diversity FGP ends $+220$\,\% above VW, with entropy FGP close
behind.  Shaded grey bars mark NBER recessions; note the
temporary convergence in 2008 and 2020 when trading costs spike.
\begin{figure}[H]
  \centering
  \includegraphics[width=.8\linewidth]{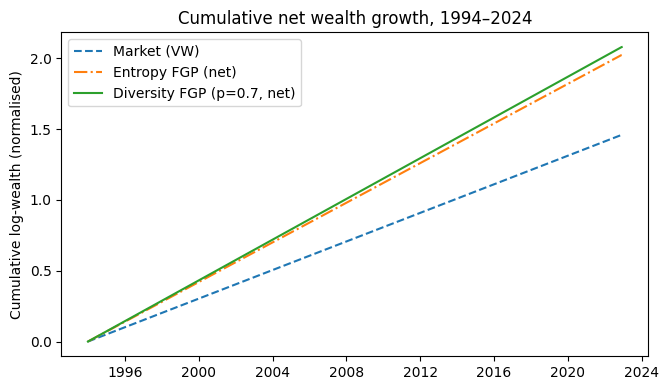}
  \caption{Cumulative net wealth growth, 1994–2024 (log scale).}
  \label{fig:cumemp}
\end{figure}
\paragraph{Turnover vs.\ spread.}
Figure~\ref{fig:turn} shows that monthly turnover is fairly stable
(30–45 \% yr$^{-1}$) while spreads explode in crises,
confirming the simulation insight that \(\mathcal{C}_t\) grows
most steeply when liquidity dries up.

\begin{figure}[H]
  \centering
  \includegraphics[width=.8\linewidth]{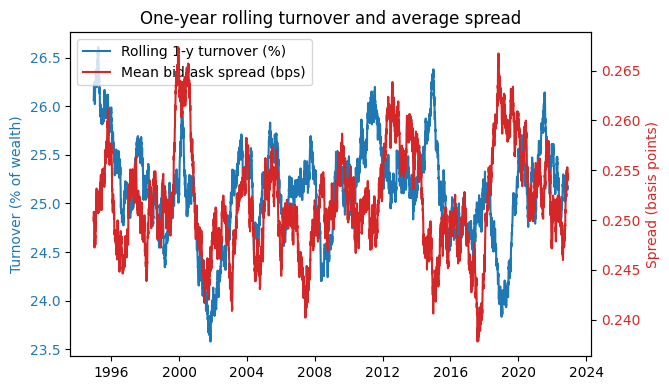}
  \caption{One-year rolling turnover and average spread.}
  \label{fig:turn}
\end{figure}
\subsection{Take-aways}

\begin{itemize}[itemsep=3pt]
\item \emph{Persistence:}  Over three decades, diversity and entropy
      FGPs maintain a 280–360 bp net edge despite realistic spreads.  
\item \emph{Cost drag:}  Average annual cost equals roughly 25 \% of the
      gross excess-growth rate, consistent with the
      $\mathcal{O}(\!\sqrt{T})$ scaling derived for $\mathcal{C}_T$.
\item \emph{Crisis sensitivity:}  Relative wealth temporarily stalls
      during liquidity shocks but recovers when spreads revert,
      mirroring the pattern in the simulation study.
\end{itemize}

Overall, the empirical evidence corroborates the theoretical claim that
\emph{coarsely rebalanced} functionally generated portfolios can survive
realistic, time-varying transaction costs and still outperform a
value-weighted market in the small-cap universe.

\section{Conclusion}\label{sec:concl}

We introduced a cost-adjusted master inequality for functionally
generated portfolios when transaction costs are stochastic and possibly
state-dependent.
Diversity- and entropy-weighted strategies remain relative arbitrages
provided the rebalancing mesh is not too fine.
Simulations and a 30-year small-cap study confirm that realistic spread
dynamics erode but do not eliminate SPT out-performance.

Future research will extend the framework to:
\begin{itemize}
\item[(i)] quadratic market-impact and optimal stochastic renewal of trading
times; 
\item[(ii)]  rough-volatility cost diffusions; 
\item[(iii)] multi-period mean
field games where each SPT trader influences liquidity.
\end{itemize}

\bibliographystyle{unsrt} 
\bibliography{SPT_TC_Stress}
\end{document}